\documentclass[submission,copyright,creativecommons]{eptcs}
\usepackage{amssymb}
\providecommand{\event}{TARK 2017}
\usepackage{underscore}
 
\input{tcilatex}

\title{Bayesian Decision Theory and Stochastic Independence}
\def\titlerunning{Bayesian Decision Theory and Stochastic Independence}
\def\authorrunning{Philippe Mongin}

\author{Philippe Mongin \institute{CNRS \& HEC Paris}\email{mongin@greg-hec.com}}

\begin{document}

\maketitle

\begin{abstract}
Stochastic independence has a complex status in
probability theory. It is not part of the definition of a probability
measure, but it is nonetheless an essential property for the mathematical
development of this theory. Bayesian decision theorists such as Savage can
be criticized for being silent about stochastic independence. From their
current preference axioms, they can derive no more than the definitional
properties of a probability measure. In a new framework of twofold
uncertainty, we introduce preference axioms that entail not only these
definitional properties, but also the stochastic independence of the two
sources of uncertainty.\ This goes some way towards filling a curious lacuna
in Bayesian decision theory.
\end{abstract}

\section{Introduction and preview}

The property of stochastic (or statistical) independence occupies a rather
special place in the mathematical theory of probability.\ It does not belong
to the properties that this theory singles out to define a probability
measure axiomatically.\ It is indeed a property of given events for a given
probability measure, and\ its adoption can only result from a modelling
choice to fit the particular situation. At the same time, probability theory
obviously uses independence assumptions extensively; they are needed for
such major results as the Laws of Large Numbers, various theorems on
stochastic processes, and some central results of statistical theory. For
Kolmogorov \cite{Kolmo56} himself, the inventor of the axiomatic definition,
this property occupies a "central position in the theory of probability"
(1933-1950, p.\ 8). One would thus expect all theories of the foundations of
probability to pay careful attention to stochastic independence, but
curiously, this is not the case with Bayesian decision theory, one of the
most influential among these theories.

Bayesian decision theorists claim that an agent's uncertain beliefs should
be represented be a probability measure and ground this claim on a pragmatic
argument.\ They formally show that if the agents' preferences over uncertain
prospects - typically, but not exclusively over monetary bets - obey certain
requirements of practical rationality, these agents' beliefs should conform
to the axiomatic definition of a probability measure.\ Bayesian decision
theorists hardly go beyond this demonstration, and in particular have
nothing to add on stochastic independence.\ Thus, they can be criticized for
falling short of justifying the probability calculus as it actually works
and stopping too early in their foundational work.

More technically, Bayesian decision theorists prove a representation theorem
for preferences over uncertain prospects that involves two sets of
quantities, utilities (over the consequences of prospects) and probabilities
(over the uncertain events), these two items being combined by the familiar
rule of expected utility (EU). After\ Ramsey's and de Finetti's sketches,
this strategy was implemented in full detail by Savage (1954) \cite{Sava54}%
.\ In a subsequent simplification of Savage's system, Anscombe and Aumann
(1963) \cite{AnAu63} took some probability values for granted in order to
obtain the remaining ones more easily. All these authors derive a \textit{%
prior} probability measure to represent initially uncertain beliefs.\ Savage
(1954-1972, p. 44) extends this argument to obtain a \textit{posterior}
probability measure, i.e., one that represents beliefs after a partial
resolution of uncertainty, and he shows that this posterior obeys Bayes's
rule of revision; literally, the "Bayesian" label becomes fully justified
only at this stage. This is also where Savage stops.\ He however
acknowledges that a treatment of stochastic independence should have come
next.

Two comments in Savage make this point clearly.\ Having axiomatized a
qualitative probability relation, he complains that "the notions of
independence and irrelevance have ... no analogues in qualitative
probability; this is surprising and unfortunate, for these notions seem to
evoke a strong intuitive response" (1954-1972, p.\ 44). Later, he reiterates
the complaint differently: "it would be desirable, if possible, to find a
simple qualitative personal description of independence between events" (p.\
91). (Savage prefers the expression of "personal probability" to the more
received one of "subjective probability".) In today's Bayesian theory, the
first comment is not justified anymore. There now exist richer systems of
qualitative probability than Savage's, in which a special relation serves to
express the stochastic independence of two events or two random variables
(see Domotor, 1969 \cite{Domo69}, Fine, 1971 \cite{Fine73}, Kaplan and Fine,
1977 \cite{KaFi77}, Luce and Nahrens, 1978 \cite{LuNa78}, to cite but the
early papers). However, the second comment is still topical.\ We understand
it as referring to preferences over uncertain prospects, i.e., the ultimate
primitive in Savage's construction. To the best of our knowledge, Bayesian
decision theory has not yet explicated stochastic independence in terms of
this overarching concept. \ 

The present paper is an attempt to do so. We assume that there are \textit{%
two distinctive sources of uncertainty}, and accordingly that states of
nature have the form of two-component vectors. As in Savage and in Anscombe
and Aumann, we define uncertain prospects to be mappings from states of the
world to consequences, and take the agent's preferences over these prospects
to be the only axiomatic primitive.\ Our construction leans towards Anscombe
and Aumann by considering a finite number of states and by making a
structural assumption on consequences (they are real numbers). However, it
also leans towards Savage because we eschew any numerical data, hence
Anscombe and Aumann's questionable trick of taking some probability values
for granted. The twofold uncertainty framework has recently been introduced
by Mongin and Pivato (2016) \cite{MoPi16} with a theoretical purpose
different from the present one.\footnote{%
That paper aims at offering a solution to the classic problem of defining a
normatively compelling notion of social preference under uncertainty.}

Our axioms entail that there exists an EU representation for the agent's
preferences, and that the probability measure in this representation
decomposes \textit{multiplicatively} on the two sources of uncertainty,
which establishes their stochastic independence.\ The same uniqueness
conditions hold as in standard EU representation theorems. A heuristic
argument indicates where to locate the stochastic independence property in
our preference axioms.\ They state that, for either source of uncertainty,
there exist preferences conditional on each value this source can take, and
moreover that these conditional preferences are invariant across possible
values.\ This heuristically means that the realization of one of the two
uncertainty components does not affect the agent's preferences over those
prospects which only depend on the other, still unknown component.\ In a
betting interpretation, if the initial bets relate, say, to tomorrow's
weather and tomorrow's economic conditions, and the agent somehow comes to
know what tomorrow's weather will be, this does not affect the agent's
preferences amongst bets on tomorrow's economic conditions, and vice-versa.\ 

We offer two representation theorems along the lines just explained. The
first adapts a result already obtained in Mongin and Pivato (2015) \cite%
{MoPi15}. We develop it here because it neatly exemplifies how Bayesian
decision theory can approach stochastic independence. Although this theorem
implements the heuristics of last paragraph, it is, in a subtle sense to be
explained, not yet entirely satisfactory.\ Hence we propose a second
representation theorem, which is the technical novelty of this paper.
Section 2\ adds further motivations.\ Sections 3 and 4 state the two
representation theorems respectively.\ Section 5 returns to conceptual
comments and comparisons, while also sketching directions for future work.

\section{Further motivating the approach}

A corn producer must decide how much land to cultivate while not knowing
what the climatic conditions and the state of demand for corn will be at the
time of the harvest. Each cultivation policy can be analyzed as an uncertain
prospect, i.e., a mapping from the unknown states to the possible
consequences, here monetary proceeds.\ We will develop this example along a
Bayesian theorist's line, and heuristically reason backwards, taking for
granted what the Bayesian theorist would conclude, plus the target property
of stochastic independence.\ Recall the textbook definition since
Kolmogorov: given a probability space $(\Omega ,\mathcal{A},P)$, two events $%
A,B$ in $\mathcal{A}$ are said to be stochastically (or statistically)
independent if $P(A\cap B)=P(A).P(B)$. From this definition, others, which
are equally standard, follow concerning collections of events or random
variables.

Assuming for simplicity that climate and demand take two values, we fix two
sets $S=\left\{ s_{1},s_{2}\right\} $ and $T=\left\{ t_{1},t_{2}\right\} $
and define a state of nature to be any element of the product set $S\times T$%
. We now apply stochastic independence to each possible pair of events
(subsets) $\left\{ s\right\} \times T$ and $S\times \left\{ t\right\} $, or
by an obvious identification, to each possible pair $(s,t)$. The producer's
probabilities are thus given by the matrix:

$%
\begin{array}{ccc}
& t_{1} & t_{2} \\ 
s_{1} & p_{s_{1}}q_{t_{1}} & p_{s_{1}}q_{t_{2}} \\ 
s_{2} & p_{s_{2}}q_{t_{1}} & p_{s_{2}}q_{t_{2}}%
\end{array}%
$\noindent

\noindent where $(p_{s_{1}},p_{s_{2}})$ and $(q_{t_{1}},q_{t_{2}})$ are
probability vectors on $S$ and $T$, respectively. Now, a policy $\mathbf{X}$
for the producer can be represented by its monetary proceeds in the various
states:

$%
\begin{array}{ccc}
\mathbf{X} & t_{1} & t_{2} \\ 
s_{1} & x_{11} & x_{12} \\ 
s_{2} & x_{21} & x_{22}%
\end{array}%
$

\noindent Denote by $\succsim $ the producer's preferences over policies $%
\mathbf{X}$.\ The EU representation for $\succsim $ is

\[
V(\mathbf{X}%
)=p_{s_{1}}q_{t_{1}}u(x_{11})+p_{s_{1}}q_{t_{2}}u(x_{12})+p_{s_{2}}q_{t_{1}}u(x_{21})+p_{s_{2}}q_{t_{2}}u(x_{22})%
\text{.} 
\]%
This can be restated either as: \ 
\[
(\ast )\text{ }V(\mathbf{X})=p_{s_{1}}\left[
q_{t_{1}}u(x_{11})+q_{t_{2}}u(x_{12})\right] +p_{s_{2}}\left[
q_{t_{1}}u(x_{21})+q_{t_{2}}u(x_{22})\right] \text{,} 
\]%
or as:

\[
(\ast \ast )\text{ }V(\mathbf{X})=q_{t_{1}}\left[
p_{s_{1}}u(x_{11})+p_{s_{2}}u(x_{21})\right] +q_{t_{2}}\left[
p_{s_{1}}u(x_{12})+p_{s_{2}}u(x_{22})\right] \text{.} 
\]

The bracketed sums in $(\ast )$ contain utility representations for
conditional preferences on the possible values of $s$, and those in $(\ast
\ast )$ contain utility representations for conditional preferences on the
possible values of $t$. Thus, the overall conclusions entail that (i) 
\textit{conditional preferences are orderings}.\ Since the same functional
form $q_{t_{1}}u(.)+q_{t_{2}}u(.)$ appears in the two bracketed sums of $%
(\ast )$, and similarly, the same functional form $%
p_{s_{1}}u(.)+p_{s_{2}}u(.)$ appears in the two bracketed sums of $(\ast
\ast )$, these conclusions also entail that (ii) \textit{conditional
preferences are the same for different }$s$\textit{, and the same for
different }$t$. Lastly, from the same equations, if the conditional
orderings for both $s_{1}$ and $s_{2}$, or the conditional orderings for
both $t_{1}$ and $t_{2}$, agree to rank prospect $\mathbf{X}$ above prospect 
$\mathbf{Y}$, then the overall preference $\succsim $ ranks $\mathbf{X}$
above $\mathbf{Y}$.\ Thus, the conclusions also entail that (iii) \textit{%
preferences over prospects are increasing with respect to either family of
conditional preferences}.

Importantly, we have stated (i), (ii) and (iii) by abstracting from the EU
representation.\ Each of these properties can indeed be satisfied by more
general theories than Bayesian decision theory, and in particular, the 
\textit{dominance} property (iii) is well-known to apply to most existing
alternatives (like rank-dependent theory, see, e.g., Wakker, 2010 \cite%
{Wakk10}).

In the first result, we assume (i), (ii) and (iii), plus some background
conditions.\ Given the formal definition of a conditional, which is restated
below, it is actually possible to fuse (i) with (iii) and obtain an even
more condensed system.\ One may wonder how apparently weak necessary
conditions for the representation turn out also to be sufficient for it.\
The key point is that the conditions apply to $s$ and $t$ \textit{at the
same time}, and this creates the possibility of representing the preference $%
\succsim $ both in terms of $s$-conditionals and $t$-conditionals; comparing
these representations leads to the results.\ Their equivalence shows in the
fact that either the $p_{s}$ or the $q_{t}$ can be factored out from the
same sum - see $(\ast )$ and $(\ast \ast )$.

\section{A first representation theorem for stochastic independence}

Formally, there are two variables of interest, $s\in S$ and $t\in T$, and a
state of the world is any pair $(s,t)$ $\in $\ $S\times T$; we thus permit
the two variables to vary together in any possible way. For technical
reasons, we take $S$ and $T$ to be finite with cardinalities $\left\vert
S\right\vert $, $\left\vert T\right\vert \geq $ $2$. Prospects $\mathbf{X}$
are mappings from states $(s,t)$ to real numbers $x$, and we define the set
of prospects to be $%
\mathbb{R}
^{S\times T}$, thus putting no constraint on what counts as a prospect. The
sets of all probability functions on $S$, $T$ and $S\times T$ are denoted by 
$\Delta _{S}$, $\Delta _{T}$ and $\Delta _{S\times T}$, respectively.

It is convenient to represent prospects $\mathbf{X}$ as $\left\vert
S\right\vert \times \left\vert T\right\vert $ matrices, with each $s$
standing for a row and each $t$ standing for a column. We will thus write $%
\mathbf{X}=\left[ x_{s}^{t}\right] _{s\in S}^{t\in T}$ , but sometimes also $%
\mathbf{X}=(\mathbf{x}_{1}$,...,$\mathbf{x}_{\left\vert S\right\vert })$,
where each component is a row vector $\mathbf{x}_{s}\in 
\mathbb{R}
^{T}$, or $\mathbf{X}=(\mathbf{x}^{1}$,...,$\mathbf{x}^{\left\vert
T\right\vert })$, where each component is a column vector $\mathbf{x}^{t}\in 
\mathbb{R}
^{S}$.

By assumption, the agent compares prospects in terms of an \textit{ex ante}
preference relation $\succsim $.\ As a maintained assumption, we take this
relation to be a continuous weak ordering, hence representable by a
continuous utility function. The other preference relations are obtained
from $\succsim $ as conditionals.\ There are three families of conditionals
to consider, i.e., $\left\{ \succsim _{s}\right\} _{s\in S}$, $\left\{
\succsim _{t}\right\} _{t\in T}$ and $\left\{ \succsim _{st}\right\} _{s\in
S,t\in T}$.\ The\ last family\ represents \textit{ex post} preferences, and\
the first two represent \textit{interim\ }preferences, since each relation
in these families depends on fixing one variable and letting the other vary,
and this amounts to resolving only part of the uncertainty.

We now formally define the various conditionals in terms of the master
relation $\succsim $. \ The \textit{conditional} \textit{of} $\succeq $ 
\textit{on} $s\in S$ is the relation $\succeq _{s}$ on $%
\mathbb{R}
^{T}$ defined by the property that for all $\mathbf{x}_{s},\mathbf{y}_{s}\in 
$ $%
\mathbb{R}
^{T}$, 
\begin{eqnarray*}
\mathbf{x}_{s}\text{ } &\succeq &_{s}\text{ }\mathbf{y}_{s}\text{\ iff }%
\mathbf{X\succsim Y} \\
\text{for some }\mathbf{X,Y} &\in &%
\mathbb{R}
^{S\times T}\text{ s.t. }\mathbf{x}_{s}\text{ is the }s\text{-row of }%
\mathbf{X}\text{, }\mathbf{y}_{s}\text{ is the }s\text{-row of }\mathbf{Y}%
\text{, } \\
&&\text{and }\mathbf{X}\text{ and }\mathbf{Y}\text{ are equal outside their }%
s\text{-row. }
\end{eqnarray*}%
Similarly, the \textit{conditional of} $\succeq $ \textit{on} $t\in T$ is
the relation $\succsim _{t}$ on $%
\mathbb{R}
^{S}$ defined by the property that for all $\mathbf{x}^{t},\mathbf{y}^{t}\in 
$ $%
\mathbb{R}
^{S}$ , 
\begin{eqnarray*}
\mathbf{x}^{t}\text{ } &\succsim &_{t}\text{ }\mathbf{y}^{t}\text{\ iff }%
\mathbf{X\succsim Y} \\
\text{for some }\mathbf{X,Y} &\in &%
\mathbb{R}
^{S\times T}\text{ s.t. }\mathbf{x}^{t}\text{ is the }t\text{-column of }%
\mathbf{X}\text{, }\mathbf{y}^{t}\text{ is the }t\text{-column of }\mathbf{Y}%
\text{,} \\
&&\text{and }\mathbf{X}\text{ and }\mathbf{Y}\text{ are equal outside their }%
t\text{-column. }
\end{eqnarray*}

\noindent By themselves, these definitions do not make conditionals weak
orderings. By a well-known fact of decision theory, $\succeq _{s}$ is a 
\textit{weak ordering} if and only if the choice of $\mathbf{X},\mathbf{Y}$
in the definition of $\succeq _{s}$ is immaterial, or more precisely, if and
only if $\mathbf{X\succsim Y\Longleftrightarrow }$ $\mathbf{X}^{\prime }%
\mathbf{\succsim Y}^{\prime }$ when $\mathbf{X}^{\prime }\mathbf{,Y}^{\prime
}$ also satisfy the condition stated for $\mathbf{X},\mathbf{Y}$ in this
definition. When this holds, $\succeq $ \ is said to be \textit{weakly
separable} in $s$.\ By another well-known fact, weak separability in a
factor (or set of factors) is equivalent to the property that $\mathbf{%
\succsim }$ is \textit{increasing} with the conditional on this factor (or
the conditionals of the set of factors). That is to say, for all $\mathbf{X,Y%
}\in 
\mathbb{R}
^{S\times T}$, if $\mathbf{x}_{s}$ $\succeq _{s}$ $\mathbf{y}_{s}$\ for all $%
s$, then $\mathbf{X\succsim Y}$; and if moreover $\mathbf{x}_{s}$ $\succ
_{s} $ $\mathbf{y}_{s}$\ for some $s$, then $\mathbf{X\succ Y}$.\footnote{%
By $\succ $, we mean the \textit{strict} preference associated with the weak
preference $\succeq $, and similarly for for $\succ _{s}$, $\succ ^{t}$ and $%
\succ _{st}$.} Everything said for $s$ of course applies to $t$. Combining
the two well-known facts, we see that conditions (i) and (iii) of the
previous section can be fused into the single requirement that all $\succeq
_{s}$ and all $\succeq _{t}$ are weak orderings.\footnote{%
For these definitions and basic facts, see Fishburn (1970) \cite{Fish70},
Keeney and Raiffa (1976) \cite{KeRa76}, and Wakker (1989) \cite{Wakk89}.}

The \textit{conditional of} $\succeq $ \textit{on} $(s,t)\in S\times T$ is
defined similarly. Since this conditional $\succsim _{st}$ compares real
numbers, it makes sense to identify it with the natural order of these
numbers. This amounts to saying that numbers represent desirable quantities,
be they money values, as in the producer example, or something else. Thus,
as another maintained assumption, we require that for all $(s,t)\in S\times
T $ and all $x_{s}^{t},y_{s}^{t}$ $\in 
\mathbb{R}
$,

\[
x_{s}^{t}\succsim _{st}y_{s}^{t}\text{ iff }x_{s}^{t}\geq y_{s}^{t}\text{.} 
\]

\noindent Since this turns the $\succsim _{st}$ into an ordering, $\succsim $
is increasing with each of these conditionals, hence also with each entry $%
x_{s}^{t}$ of $\mathbf{X}$.

Let us say that the family $\left\{ \succsim _{s}\right\} _{s\in S}$ ( $%
\left\{ \succsim _{t}\right\} _{t\in T}$) is \textit{invariant} if $\succsim
_{s}=$ $\succsim _{s^{\prime }}$ for all $s,s^{\prime }\in S$ (resp. $%
\succsim _{t}=$ $\succsim _{t^{\prime }}$ for all $t,t^{\prime }\in T$).\
Such requirements capture condition (ii) of previous section. Notice they
are not needed for the $\succsim _{st}$ since these are identical relations
by construction.

We are now ready for a representation theorem.

\bigskip

\begin{theorem}
The following conditions are equivalent:

\begin{itemize}
\item The conditionals $\succsim _{s}$ and $\succsim _{t}$ are weak
orderings for all $s\in S$ and all $t\in T$, and each family of conditionals
is invariant. \ 

\item There are an increasing, continuous function $u:{\mathcal{%
\mathbb{R}
}\longrightarrow }{\mathbb{R}}$, and strictly positive probability functions 
${\mathbf{p}}\in \Delta _{S}$ and $\mathbf{q}\in \Delta _{T}$, such that $%
\succeq $ is represented by the function $V:\mathcal{%
\mathbb{R}
}^{S\times T}{\longrightarrow }{\mathbb{R}}$ that computes the ${\mathbf{%
p\otimes q}}$-expected value of $u$, i.e., by the function defined as
follows: for all $\mathbf{X}=\left[ x_{s}^{t}\right] _{s\in S}^{t\in T}$ , 
\[
V({\mathbf{X}})\ :=\ \ \sum_{s\in S}\sum_{t\in T}p_{s}\,q_{t}u(x_{s}^{t})%
\text{.}\ 
\]%
$\ $\noindent 
\end{itemize}

\noindent In this format of EU\ representation, ${\mathbf{p}}$ and $\mathbf{q%
}$ are unique, and $u$ is unique up to positive affine transformations.
\end{theorem}

As was foreshadowed, the representation theorem combines the conclusions of
Bayesian decision theory with the desired property that the probability
measure (here a vector) is multiplicative in the two sources. This theorem
is Corollary 1(c) in Mongin and Pivato (2015) \cite{MoPi15}, but recast
autonomously and in a different formal language, so as to facilitate the
discussion of stochastic independence.

\section{A second representation theorem for stochastic independence}

In Theorem 1, strong results follow from a compact list of assumptions,
undoubtedly a feature of mathematical elegance, but also a cause for
conceptual dissatisfaction.\ Would it not be better to expand on the
assumptions and separate those which are responsible for the existence of
the EU representation and those which account for the stochastic
independence property occurring in this representation?\ This disentangling
would make sense on two counts: stochastic independence is an optional
property of probability measures (a logical point) and Bayesian decision
theory invests only the existence, not the properties, of such measures with
universal rationality significance (a normative point). However, the
assumptions of Theorem 1 cannot be divided in the appropriate way. This can
be seen as follows.\ By taking the $\succsim _{s}$ and $\succsim _{t}$ to be
merely orderings, not invariant orderings, one would get an additively
separable representation that does not separate the utility and probability
components of the added terms.\ By taking only one of the two families to
satisfy the ordering and invariance assumptions, one would only get a
representation that is only separable in that family and says nothing on
probabilities either.\footnote{%
The additively separable representation of the first case reads as 
\[
\sum_{s\in S,t\in T}v_{st}(x_{s}^{t})\text{,}
\]%
with increasing and continuous $v_{st}:%
\mathbb{R}
\rightarrow 
\mathbb{R}
$, $s\in S,t\in T$.\ In the second case, if the assumptions only hold for
the $\succsim _{s}$, the separable representation reads as 
\[
W(V_{1}(\mathbf{x}_{1}),...,V_{\left\vert S\right\vert }(\mathbf{x}%
_{\left\vert S\right\vert }))\text{,}
\]%
with increasing and continuous $W:$ $%
\mathbb{R}
^{S}\rightarrow 
\mathbb{R}
$ and $V_{s}:%
\mathbb{R}
^{T}\rightarrow 
\mathbb{R}
$, $s\in S$.\ These partial results follow from standard results in
separability theory.}

Fortunately, we can obtain a relevant partitioning of assumptions if we
enrich the decision-theoretic framework beyond the present, two-dimensional
stage.\ Let us suppose that the agent pays attention not only to the
uncertainty dimensions $s$ and $t$ of the final consequences, but also to a
third, heuristically unrelated dimension $i$, so that these consequences are
now represented by real numbers $x_{st}^{i}$.\ The added dimension can be
interpreted as being the \textit{time} at which these consequences occur,
and\ the corn producer decision problem can easily be reformulated
accordingly.\ For technical reasons, we require $i$ to take its values in a
finite set $I$ with cardinality $\left\vert I\right\vert \geq $ $2$.

Given this added dimension, alternatives become mappings from triples $%
(s,t,i)$ to the reals, that is three-dimensional arrays, $\mathbb{\ }$%
\[
\mathbb{X}=\left[ x_{s}^{t}\right] _{s\in S,t\in T}^{i\in I}\in \mathbb{R}%
^{S\times T\times I}\text{,} 
\]%
which may be rewritten as 
\[
\mathbb{X}=(\mathbf{X}_{1},...,\mathbf{X}_{\left\vert S\right\vert })\text{, 
}\mathbb{X=}(\mathbf{X}_{1},...,\mathbf{X}_{\left\vert T\right\vert })\text{
or }\mathbb{X}=(\mathbf{X}^{1},...,\mathbf{X}^{\left\vert I\right\vert })%
\text{,} 
\]%
where the components are matrix-valued, i.e., $\mathbf{X}%
_{s}=(x_{st}^{i})_{t\in T}^{i\in I}\in \mathbb{R}^{T\times I}$, $\mathbf{X}%
_{t}=(x_{st}^{i})_{s\in \mathcal{S}}^{i\in I}\in \mathbb{R}^{\mathcal{S}%
\times I}$ and $\mathbf{X}^{i}=(x_{st}^{i})_{s\in \mathcal{S},t\in T}\in 
\mathbb{R}^{\mathcal{S}\times T}$ respectively.

With $i$ representing time, alternatives should be viewed as \textit{%
contingent plans}, i.e., plans whose consequences in a given period depend
on the way the uncertainty - still represented by $(s,t)$ - is resolved in
that period. The matrix-valued objects just listed are interpreted in terms
of \textit{partly contingent plans} (when one dimension of uncertainty is
fixed) or \textit{dated prospects} (when the time dimension is fixed).
Vector-valued objects can also be interpreted - e.g., the $\mathbf{x}%
_{st}=(x_{st}^{i})^{i\in I}\in \mathbb{R}^{I}$ as \textit{non-contingent
plans}.\footnote{%
It is essential for these interpretations that each period is uncertain in
the same way as any other, i.e., no interaction exists between the
resolution of uncertainty and the passing of time.}

We assume the agent compares contingent plans in terms of a preference
relation $\succsim $, which is a continuous weak ordering, and this relation
gives rise to conditional relations that are defined as in the previous
section, \textit{mutatis mutandis}.\ Among the seven families of
conditionals, we pay special attention to $\left\{ \succsim _{s}\right\}
_{s\in S}$, $\left\{ \succsim _{t}\right\} _{t\in T}$, $\left\{ \succsim
^{i}\right\} ^{i\in I}$, $\left\{ \succsim _{st}\right\} _{s\in S,t\in T}$
and $\left\{ \succsim _{st}^{i}\right\} _{s\in S,t\in T}^{i\in I}$.\ The $%
\succsim _{s}$ and $\succsim _{t}$ compare plans $\mathbf{X}_{s}$ and $%
\mathbf{X}_{t}$, respectively; the $\succsim ^{i}$ compare dated prospects $%
\mathbf{X}^{i}$, the $\succsim _{st}$ non-contingent plans $\mathbf{x}_{st}$
and the $\succsim _{st}^{i}$real-valued consequences. As before, we assume
that each $\succsim _{st}^{i}$ coincides with the natural order of real
numbers, and that\ all other relations may or may not be weak orderings, and
may or may not form an invariant family.

\begin{theorem}
The following conditions are equivalent:

\begin{itemize}
\item The conditionals $\succsim ^{i}$ and $\succsim _{st}$ \ are weak
orderings, and the $\succsim _{st}$ family is invariant. \ 

\item There are increasing, continuous functions $u^{i}:{\mathcal{%
\mathbb{R}
}\longrightarrow }{\mathbb{R}}$, for all $i\in I$, and a strictly positive
probability function $\pi \in \Delta _{S\times T}$, such that $\succeq $ is
represented by the function $W:\mathcal{%
\mathbb{R}
}^{S\times T\times I}{\longrightarrow }{\mathbb{R}}$ that computes the $%
\mathbf{\pi }$-expected value of $\sum_{i\in I}u^{i}$, i.e., the function
thus defined: for all $\mathbb{X}=\left[ x_{s}^{t}\right] _{s\in S,t\in
T}^{i\in I}$ , 
\[
(\ast )\text{ }W(\mathbb{X})\ :=\ \ \sum_{s\in S,t\in T}\sum_{i\in I}\pi
_{st}\,u^{i}(x_{st}^{i})\text{.}\ 
\]%
$\ $\noindent
\end{itemize}

\noindent In this format of representation, $\mathbf{\pi }$ is unique, and
the $u^{i}$ are unique up to positive affine transformations with a common
multiplier.

Moreover, the following are equivalent:

\begin{itemize}
\item The assumptions made above on the $\succsim ^{i}$ and $\succsim _{st}$
\ hold, and furthermore the $\succsim _{s}$ are weak orderings and an
invariant family. \ 

\item The same conclusions hold, and moreover there are strictly positive
probability functions $\mathbf{p}\in \Delta _{S}$ and $\mathbf{q}\in \Delta
_{T}$ with $\pi =\mathbf{p\otimes q}$, so that $(\ast )$ becomes: for all $%
\mathbb{X}=\left[ x_{s}^{t}\right] _{s\in S,t\in T}^{i\in I}$ , 
\[
(\ast \ast )\text{ }W(\mathbb{X})=\ \ \sum_{s\in S}\sum_{t\in T}\sum_{i\in
I}p_{s}\,q_{t}u^{i}(x_{st}^{i})\text{.}\ 
\]%
$\ $\noindent 
\end{itemize}

\noindent In this alternative format, ${\mathbf{p}}$ and $\mathbf{q}$ are
unique, while the $u^{i}$ have the same uniqueness properties as before.
\end{theorem}

\begin{proof}
(Sketch).\ The first part follows from Theorem 1 in Mongin and Pivato (2015) 
\cite{MoPi15}; we leave it for the reader to check that the assumptions of
this theorem apply here. For the second part, we first observe that, for
every given $s\in S$, the $W(\mathbb{X})$ representation delivers a function 
$\mathbb{R}^{T\times I}\rightarrow 
\mathbb{R}
$ 
\[
\sum_{t\in T}\sum_{i\in I}\pi _{st}u^{i}(x_{st}^{i}) 
\]%
that represents the weak ordering $\succsim _{s}$.\ If we define $\pi
_{st}^{\prime }:=\pi _{st}$/$\sum_{t\in T}\pi _{st}$ for all $t\in T$, the
function 
\[
\sum_{t\in T}\sum_{i\in I}\pi _{st}^{\prime }u^{i}(x_{st}^{i}) 
\]%
is also a representation of $\succsim _{s}$. Now fix $s_{0}\in S$.\ By the
invariance of the $\succsim _{s}$ family, for every $s$ $\in S$, there is a
strictly increasing transformation $\Phi _{s}$ s.t.

\[
\sum_{t\in T}\sum_{i\in I}\pi _{s_{0}t}^{\prime }\,u^{i}(x_{st}^{i})=\Phi
_{s}\left( \sum_{t\in T}\sum_{i\in I}\pi _{st}^{\prime
}u^{i}(x_{st}^{i})\right) \text{.} 
\]%
As the $u^{i}$ are strictly increasing and continuous, $\Phi _{s}$ is $%
\mathbb{R}
\rightarrow 
\mathbb{R}
$, and we can apply a functional equation argument (Rado and Baker, 1987 
\cite{RaBa87}) and conclude that the $\Phi _{s}$ are positive affine
transformations.\ I.e., for all $s$ $\in S$, there exist numbers $\alpha
_{s}>0$ and $\beta _{s}$ s.t.

\[
\sum_{t\in T}\sum_{i\in I}\pi _{s_{0}t}^{\prime }\,u^{i}(x_{st}^{i})=\alpha
_{s}\left( \sum_{t\in T}\sum_{i\in I}\pi _{st}^{\prime
}u^{i}(x_{st}^{i})\right) +\beta _{s}\text{.} 
\]%
After redefining the functions so as to dispense with the constant terms, we
see that, for all $s$ $\in S$ and $t\in T$, $\pi _{s_{0}t}^{\prime
}\,=\alpha _{s}\pi _{st}^{\prime }$, and in fact (since proportional
probability vectors are equal) $\pi _{s_{0}t}^{\prime }\,=\pi _{st}^{\prime
} $.\ We thus rewrite $(\ast )$ as 
\[
\ \ \sum_{s\in \mathcal{S},t\in T}\sum_{i\in I}\pi _{s_{0}t}^{\prime
}\,(\sum_{t\in T}\pi _{st})(u^{i}(x_{st}^{i})\text{,}\ 
\]%
which is (**) if one takes $\mathbf{p=}$ $(\sum_{t\in T}\pi _{st})_{s\in S}$%
\ and $\mathbf{q}=$ $(\pi _{s_{o}t}^{\prime })_{t\in T}$. The uniqueness of $%
\mathbf{p}$ and $\mathbf{q}$ in this format of representation is easily
established.
\end{proof}

\bigskip

The two steps of Theorem 2 correspond to the EU representation theorem and
the addition made by stochastic independence, respectively.\ Interestingly,
the same assumption - that of invariant conditional orderings - underlies
both conclusions, but in the richer framework adopted here, it is possible
to apply it twice over, thus separating each step. Note also that it is
enough to apply the assumption to the $\succsim _{st}$ and one of the two $%
\succsim _{s}$ and $\succsim _{t}$ families; then, as the representation
shows, the other family automatically satisfies this assumption (the $%
\succsim _{t}$ can of course be interchanged with the $\succsim _{s}$ in the
theorem statement).

\bigskip

\section{Interpretations, comparisons and future directions}

The following heuristic argument will help locate the preference ancestor of
stochastic independence more precisely. Considering only four states for
simplicity, we suppose that the agent considers $(s_{1},t_{1})$ more likely
that $(s_{1},t_{2})$, and $(s_{2},t_{1})$ less likely than $(s_{2},t_{2})$.\
That is, from knowing how the uncertainty on $s$ is resolved, the agent
draws an inference on how the uncertainty on $t$ would be resolved.\ If the
agent reasoned probabilistically, the joint probabilities would of course 
\textit{not} decompose multiplicatively. We now check that the $s$%
-conditional preferences cannot be invariant.\ Take $\xi ,\xi ^{\prime }$
representing desirable quantities, with $\xi >\xi ^{\prime }$, and the
following prospects in matrix form:

$%
\begin{array}{ccc}
\mathbf{X} & t_{1} & t_{2} \\ 
s_{1} & \xi & \xi ^{\prime } \\ 
s_{2} & \xi ^{\prime } & \xi%
\end{array}%
$ and $%
\begin{array}{ccc}
\mathbf{Y} & t_{1} & t_{2} \\ 
s_{1} & \xi ^{\prime } & \xi \\ 
s_{2} & \xi & \xi ^{\prime }%
\end{array}%
$ .

\noindent The first line of $\mathbf{X}$, which puts the best consequence on
the more likely state, should be preferred to the first line of $\mathbf{Y}$%
, which puts it on the less likely state.\ By a similar comparison, the
second line of $\mathbf{X}$ should be preferred to the second line of $%
\mathbf{Y}$. Thus, the two $s$-conditional preferences differ.\ Contraposing
the argument, we see that invariant $s$-preferences express a one-way form
of \textit{informational independence}, i.e., that observing $s$ does not
bring any information on $t$.\ The opposite one-way form, i.e., that
observing $t$ does not bring any information on $s$, would be expressed by\
invariant $t$-preferences.\ The two invariance conditions together convey a
sense of \textit{mutual informational independence}.\ 

This concept of independence actually underlies one of the two main informal
explications of stochastic independence, the other relying on the very
different concept of \textit{mutual causal independence}\footnote{%
The causal explication of stochastic independence is actually the older and
more established of the two.\ On the problematic linkage of stochastic
independence with causality, see Spohn (1980) \cite{Spoh80}.}.\ Today's
probability texts sometimes entangle the two independence concepts, and it
is perhaps one virtue of the above axiomatization to bring out the
informational concept in a way that precludes any confounding with its
causal competitor.\ Our betting agents may or may not be influenced by what
they perceive of causal relations, but this is irrelevant, given that only
their decisions matter to the analysis.\ However, the informational concept
and associated informal explication can be pursued in various ways, and our
axiomatization more specifically contributes to giving both a pragmatic
slant. Generally, when it comes to the epistemic aspects of probability, a
divide appears between theorists who would like to explore these aspects 
\textit{per se}, and Bayesian decision theorists, who absorb them into their
practical rationality concerns.\footnote{%
Joyce (1998) \cite{Joyc98} has made this divide very clear, while adopting
the former position.}

The introduction sketched a comparison with two such theorists that can now
be made precise. We share with Anscombe and Aumann (1963) \cite{AnAu63} not
only the two assumptions of a finite set of states and a structured set of
consequences, but also that of two distinctive sources of uncertainty.
However, unlike them, we do not suppose that one of these sources already
has a probabilistic representation, a question-begging assumption from the
perspective of Bayesian decision theory, since this would require a
preference derivation for \textit{every} kind of probability. Our completely
preference-based approach likens it to Savage's (1954) \cite{Sava54} despite
the technical dissimilarities concerning the set of states and set of
consequences, as well as an axiomatic difference we now clarify. The
assumptions in Theorems 1 and 2 that certain conditionals are orderings
amount to replacing his postulate P2, i.e., the notorious "sure-thing
principle", by a dominance principle, which is weaker and more generally
accepted.\ However, to make good for this loss, we had to revise his other
postulates, and this was done by the way in which we state the invariance of
the orderings in question.\ Savage's P3 requires that conditional
preferences be invariant across all possible conditioning events, but only
when these preferences compare constant prospects. We require invariance
only for some events, but - crucially - for any comparison of prospects,
whether these are constant or not. The other event-invariance condition of
Savage, P4, has no role to play here, because it serves to order an
unstructured consequence set, while ours\ inherits the order structure of
real numbers.\footnote{%
The difference with earlier results in Bayesian decision theory is also one
of mathematical techniques.\ The proofs of Theorems 1\ and 2 depend on
additive separability arguments; see Mongin and Pivato (2015 \cite{MoPi15},
2016 \cite{MoPi16}), and for the source papers, Debreu (1960) \cite{Debre60}
and Gorman (1968) \cite{Gorm68}.}

The results of this paper may be extended in several directions.\ One of
them is \textit{conditional probability}, and the corresponding definition
of stochastic independence in terms of this concept rather than that of
joint probability.\ Conditional prospects can be introduced into the
preference apparatus with relevant preference axioms.\ Such a variation is
unlikely to make much difference to the conclusions, but it would be judged
preferable by those probability theorists and philosophers who, unlike
Kolmogorov, regard conditional probability as the genuine primitive of the
probability calculus.\footnote{%
This is a common line to take, especially among philosophers of probability;
in connection with stochastic independence, see Fitelson and Hajek (2014) 
\cite{FiHa14}.} Another, conceptually more problematic direction is \textit{%
one-sided stochastic independence}.\ Such a concept appears rarely, if ever,
in probability theory.\ In our preference framework, it is possible to
express the idea that $s$ does not bring any information on $t$ whereas $t$
may bring information on $s$; it is enough to assume that the $s$%
-conditionals are invariant, while not assuming that the $t$-conditionals
are. However, Theorem 2 teaches us in effect that, if the preference
assumptions endow $(s,t)$ with a probabilistic representation, one-sided
informational independence automatically entails mutual informational
independence. This suggests that the issue of one-sided stochastic
independence may be impossible to pursue in the present framework; which
other preference framework would facilitate its investigation is unclear.
Last but not least, stochastic independence has been reconsidered in the
currently active work on multiple (or "imprecise") probabilities, and it
would be an interesting project to connect one or more of the definitions
given to it in this work with a decision-theoretic apparatus; the latter
would of course not be Bayesian in the usual sense. Some steps have been
taken in this direction, but much work still remains to be done.\footnote{%
See in particular Ghirardato (1997) \cite{Ghir97} and Bade (2008) \cite%
{Bade08}. Cozman (2012) \cite{Cozma12} surveys different ways in which
stochastic independence can be redefined when beliefs are represented by
sets of probabilities.}

\bigskip

\textbf{Acknowledgements.}\ Many thanks for conceptual and technical
comments to Lorraine Daston, Marcus Pivato, the audience of a seminar at the
Munich Center for Mathematical Philosophy, and three TARK\ referees.\ The
author also gratefully acknowledges the hospitality of the Max Planck
Institut f\"{u}r Wissenschaftsgeschichte zu\ Berlin when he developed this
project.

\bibliographystyle{eptcs}
\bibliography{si}

\end{document}